\documentclass[11pt,onecolumn,journal,draftclsnofoot,doublespace]{IEEEtran}
\makeatletter
\def\ps@headings{%
\def\@oddhead{\mbox{}\scriptsize\rightmark \hfil \thepage}%
\def\@evenhead{\scriptsize\thepage \hfil \leftmark\mbox{}}%
\def\@oddfoot{}%
\def\@evenfoot{}}
\makeatother
\usepackage{graphicx,cite}
\usepackage[cmex10]{amsmath}
\usepackage{amssymb}
\usepackage{array}
\usepackage{fixltx2e}
\usepackage{amsfonts}
\usepackage{lineno}
\usepackage{times,amsmath,amssymb,graphicx,epsf,psfrag,epsfig,cite,color}

\def\boxit#1{\vbox{\hrule\hbox{\vrule\kern3pt
        \vbox{\kern3pt#1\kern3pt}\kern3pt\vrule}\hrule}}

\def\reals{ { {\rm  I \kern-0.15em R }  } }
\def\complex{ {\,{{\rm C} \kern-0.50em \raise0.20ex {  |}}\, }}

\def\Rbf{{\bf R}}

\def\Ac{{\cal A}}

\def\Pc{{\cal P}}

\def\Rc{{\cal R}}

\def\Uc{{\cal U}}

\def\be{\begin{equation}}
\def\ee{\end{equation}}

\def\defeq{{\stackrel{\Delta}{=}}}

\def\Rxx{\Rbf_{\ssstyle X\kern-.1em X}}

\let\ssstyle=\scriptscriptstyle

\def\eg{{\it e.g.,\ \/}}

\def\ie{{\it i.e.,\ \/}}
\def\Kout{\setbox1=\hbox{\Huge\bf K}\hbox to
1.05\wd1{\hspace{.05\wd1}
}}
\def\Sout{\setbox1=\hbox{\Huge\bf S}\hbox to 1.05\wd1{\hspace{.05\wd1}
}}

\def\ie{{\it i.e.,\ \/}}

\def\defeq{{\,\stackrel{\Delta}{=}}\,}

\newtheorem{theorem}{Theorem}

\newtheorem{definition}{Definition}

\begin{document}
\title{\huge Adaptive Shortest-Path Routing under Unknown and Stochastically Varying Link
States}

\author{\IEEEauthorblockN{Keqin Liu\IEEEauthorrefmark{0},
Qing Zhao\IEEEauthorrefmark{0}\\
 }
 \IEEEauthorblockA{\IEEEauthorrefmark{0}
Dept.~of Elec.~and Comp.~Eng., Univ. of California, Davis\\ Email:
\{kqliu,qzhao\}@ucdavis.edu}}

\maketitle

\begin{abstract}

We consider the adaptive shortest-path routing problem in wireless
networks under unknown and stochastically varying link states. In
this problem, we aim to optimize the quality of communication
between a source and a destination through adaptive path selection.
Due to the randomness and uncertainties in the network dynamics, the
quality of each link varies over time according to a stochastic
process with unknown distributions. After a path is selected for
communication, the aggregated quality of all links on this path (\eg
total path delay) is observed. The quality of each individual link
is not observable. We formulate this problem as a multi-armed bandit
with dependent arms. We show that by exploiting arm dependencies, a
regret polynomial with network size can be achieved while
maintaining the optimal logarithmic order with time. This is in
sharp contrast with the exponential regret order with network size
offered by a direct application of the classic MAB policies that
ignore arm dependencies. Furthermore, our results are obtained under
a general model of link-quality distributions (including
heavy-tailed distributions) and find applications in cognitive radio
and ad hoc networks with unknown and dynamic communication
environments.
\end{abstract}

\section{Introduction}

We consider a wireless network where the quality state of each link
is unknown and stochastically varying over time. Specifically, the
state of each link is modeled as a random cost (or reward) evolving
i.i.d. over time under an unknown distribution. At each time, a path
from the source to the destination is selected as the communication
route and the {\em total end-to-end cost}, given by the sum of the
costs of all links on the path, is subsequently observed. The cost
of each individual link is not observable. The objective is to
design the optimal sequential path selection policy to minimize the
long-term total cost.

\subsection{Stochastic Online Learning based on Multi-Armed Bandit}

The above problem can be modeled as a generalized Multi-Armed Bandit
(MAB) with dependent arms. In the classic
MAB~\cite{Robbins:52BAMS,Lai87AS,{Lai&Robbins85AAM},{AgrawalEtal95AAP},
{Auer&etal02ML},{Liu&Zhao11Allerton}}, there are $N$ independent
arms and a player needs to decide which arm to play at each time. An
arm, when played, offers a random cost drawn from an unknown
distribution. The performance of a sequential arm selection policy
is measured by regret, defined as the difference in expected total
cost with respect to the optimal policy in the ideal scenario of
known cost models where the player always plays the best arm. The
optimal regret was shown to be logarithmic with time
in~\cite{Lai&Robbins85AAM}. Furthermore, since cost observations
from one arm do not provide information about the cost of other
arms, the optimal regret grows linearly with the number of arms. We
can model the adaptive routing problem as an MAB by treating each
path as an arm. The difference is that paths are dependent through
shared links. While the dependency across paths can be ignored in
learning and the policies for the classic MAB directly apply, such a
naive approach yields poor performance with a regret growing
linearly with the number of paths, thus exponentially with the
network size (\ie the number of links in the worst case).

In this paper, we show that by exploiting the structure of the path
dependencies, a regret polynomial with network size can be achieved
while preserving its optimal logarithmic order with time.
Specifically, we propose an algorithm that achieves an $O(d^3\log
T)$ regret for all light-tailed cost distributions, where $d\le m$
is the dimension of the path set, $m$ the number of links, and $T$
the length of time horizon. We further show that a modification to
the proposed algorithm leads to a regret linear with $d$ by
sacrificing an arbitrarily small regret order with time. This result
allows a performance tradeoff in terms of the network size and the
time horizon. For example, the algorithm with a smaller regret order
in $d$ can perform better over a finite time horizon when the
network size is large. For heavy-tailed cost distributions, a regret
linear with the number of edges and sublinear with time can be
achieved. We point out that any regret sublinear with time implies
the convergence of the time-average cost to the minimum one of the
best path.

We further generalize the adaptive routing problem to stochastic
online linear optimization problems, where the action space is a
compact subset of $\Rc^d$ and the random cost function is linear on
the actions. We show that for all light-tailed cost functions,
regret polynomial or linear with $d$ and sublinear with time can be
achieved. We point out that for cases where there exists a nonzero
gap in the expected cost between the optimal action and the rest of
the action space (\eg when the action set is a polytope or finite),
the same regret orders obtained for the adaptive routing problem can
be achieved.

\subsection{Applications}

One application example is adaptive routing in cognitive radio
networks where secondary users communicate by exploiting channels
temporarily unoccupied by primary users. In this case, the
availability of each link dynamically varies according to the
communication activities of nearby primary users. The delay on each
link can thus be modeled as a stochastic process unknown to the
secondary users. The objective is to route through the path with the
smallest latency (\ie the lightest primary traffic) through
stochastic online learning.

Other applications include ad hoc networks where link states vary
stochastically due to channel fading or random contentions with
other users.

\subsection{Related Work}\label{sec:related}

The classic MAB was addressed by Lai and Robbins in 1985, where they
showed that the minimum regret has a logarithmic order with time and
proposed specific policies to asymptotically achieve the optimal
regret for several cost distributions in the light-tailed
family~\cite{Lai&Robbins85AAM}. Since Lai and Robbins's seminar
work, simpler index policies were proposed for different classes of
light-tailed cost distributions to achieve the logarithmic regret
order with time~\cite{Lai87AS,AgrawalEtal95AAP,Auer&etal02ML}.
In~\cite{Liu&Zhao11Allerton}, we proposed the DSEE approach that
achieves the logarithmic regret order with time for all light-tailed
cost distributions and sublinear regrets with time for heavy-tailed
cost distributions. All regrets established
in~\cite{Lai&Robbins85AAM,Lai87AS,AgrawalEtal95AAP,Auer&etal02ML,Liu&Zhao11Allerton}
grow linearly with the number of arms. This work builds upon the
general DSEE structure proposed in~\cite{Liu&Zhao11Allerton}. By
incorporating the exploitation of the arm dependencies into DSEE, we
show that the learning efficiency can be significant improved in
terms of network size while preserved in terms of time.

This work generalizes the previous
work~\cite{TakimotoWarmuth03JMLR,{KalaiVempala05},{GaiEtal10TR}}
that assumes fully observable link costs on a chosen path. In
particular, the adaptive routing problem under this more informative
observation model was considered in~\cite{GaiEtal10TR} and an
algorithm was proposed to achieve $O(m^4\log T)$ regret, where $m$
is the number of links in the network. The problems considered in
this paper were also studied under an adversarial bandit model in
which the cost functions are chosen by an adversary and are treated
as arbitrary bounded deterministic
quantities~\cite{AwerbuchKleinberg08JCSS}. Algorithms were proposed
to achieve regrets sublinear with time and polynomial with network
size. The problem formulation and results established in this paper
can be considered as a stochastic version of those
in~\cite{AwerbuchKleinberg08JCSS}.

For the special class of stochastic online linear optimization
problems where there exists a nonzero gap in expected cost between
the optimal action and the rest of the action space and the cost has
a finite support, an algorithm was proposed in~\cite{DaniEtal08COLT}
to achieve an $O(d^2\log^3T)$ regret given that a nontrivial lower
bound on the gap is known. The algorithm proposed in this paper
achieves a better regret in terms of both $d$ and $T$ without any
knowledge on the cost model. For the general case with
finite-support cost distributions, the regret was shown to be lower
bounded by $O(d\sqrt{T})$ and an efficient algorithm was proposed to
achieve an $O((d\ln T)^{3/2}\sqrt{T})$ regret~\cite{DaniEtal08COLT}.
Compared to the algorithm in~\cite{DaniEtal08COLT}, our algorithm
for the general stochastic online linear optimization problems
performs worse in $T$ but better in $d$.

\section{Problem Statement}

In this section, we address the adaptive routing problem under
unknown and stochastically time-varying link states. Consider a
network with a source $s$ and a destination $r$. Let $G=(V,E)$
denote the directed graph consisting of all simple paths from $s$ to
$r$. Let $m$ and $n$ denote, respectively, the number of edges and
vertices in graph $G$.

At each time $t$, a random weight/cost $W_e(t)$ drawn from an
unknown distribution is assigned to each edge in $E$. We assume that
$\{W_e(t)\}$ are i.i.d. over time for each edge $e$. At the
beginning of each time slot $t$, a path $p_n\in
\Pc~(i\in\{1,2,\ldots,|\Pc|\})$ from $s$ to $r$ is chosen, where
$\Pc$ is the set of all paths from $s$ to $r$ in $G$. Subsequently,
the {\em total end-to-end cost} $C_t(p_n)$ of the path $p_n$, given
by the sum of the weights of all edges on the path, is revealed in
the end of the slot. We point out that the individual cost on each
edge $e\in E$ is unobservable.

The regret of a path selection policy $\pi$ is defined as the
expected extra cost incurred over time $T$ compared to the optimal
policy that always selects the best path (\ie the path with the
minimum expected total end-to-end cost). The objective is to design
a path selection policy $\pi$ to minimize the growth rate of the
regret. Let $\sigma$ be a permutation on all paths such that
\[\mathbb{E}[C_t(\sigma(1))]\le\mathbb{E}[C_t(\sigma(1))]\le\mathbb{E}[C_t(\sigma(|\Pc|))].\]
We define regret
{\[\Rc^{\pi}(T)\defeq\mathbb{E}[\sum_{t=1}^T(C_t(\pi)-C_t(\sigma(1)))],\]}where
$C_t(\pi)$ denotes the total end-to-end cost of the selected path
under policy $\pi$ at time $t$.

\section{Adaptive Shortest Path Routing Algorithm}

In this section, we present the proposed adaptive shortest-path
routing algorithm. We consider both light-tailed and heavy-tailed
cost distributions.

\subsection{Light-Tailed Cost Distributions}\label{sec:pathvec}
We consider the light-tailed cost distributions as defined below.

\begin{definition}\label{def:mgf}
A random variable $X$ is light-tailed if its moment-generating
function exists, \ie there exists a $u_0>0$ such that
{\[M(u)\defeq\mathbb{E}[\exp(uX)]<\infty~~\forall~u\le|u_0|;\]}otherwise
$X$ is heavy-tailed.
\end{definition}

For a zero-mean light-tailed random variable $X$, we
have~\cite{{CharekaEtal06JM}}, {\begin{eqnarray}\label{eqn:mgfbound}
M(u)\le\exp(\zeta
u^2/2),~~~~\forall~u\le|u_0|,~\zeta\ge\sup\{M^{(2)}(u),~-u_0\le u\le
u_0\},
\end{eqnarray}}where $M^{(2)}(\cdot)$ denotes the second derivative of $M(\cdot)$
and $u_0$ the parameter specified in Definition~\ref{def:mgf}.
From~\eqref{eqn:mgfbound}, we have the following extended
Chernoff-Hoeffding bound on the deviation of the sample mean from
the true mean for light-tailed random
variables~\cite{Liu&Zhao11Allerton}.

Let $\{X(t)\}_{t=1}^{\infty}$ be i.i.d. light-tailed random
variables. Let $\overline{X_s}=(\Sigma_{t=1}^s X(t))/s$ and
$\theta=\mathbb{E}[X(1)]$. We have, for all $\delta\in [0,\zeta
u_0], a\in (0, 1/(2\zeta)]$, {\begin{eqnarray}\label{eqn:chernoff}
\Pr(|\overline{X_s}-\theta|\ge \delta)\le 2\exp(-a\delta^2s).
\end{eqnarray}}Note that the path cost is the sum of the costs of all edges on the
path. Since the number of edges in a path is upper bounded by $m$,
the bound on the moment generating function in~\eqref{eqn:mgfbound}
holds on the path cost by replacing $\zeta$ by $m\zeta$, and so does
the Chernoff-Hoeffding bound in~\eqref{eqn:chernoff}.

In the following, we propose an algorithm to achieve a regret
polynomial with $m$ and logarithmic with $T$. We first represent
each path $p_n$ as a vector $\vec{p}_n$ with $m$ entries consisting
of $0$s and $1$s representing whether or not an edge is on the path.
The vector space of all paths is embedded in a
$d$-dimensional~($d\le m$) subspace\footnote{If graph $G$ is
acyclic, then $d=m-n+2$.} of $\Rc^m$. The cost on each path $p_n$ at
time $t$ is thus given by the linear function
{\[[W_1(t),W_2(t),\ldots,W_m(t)]\cdot \vec{p}_n.\]}The general
structure of the algorithm follows the DSEE framework established
in~\cite{Liu&Zhao11Allerton} for the classic MAB. More specifically,
we partition time into an exploration sequence and an exploration
sequence. In the exploration sequence, we sample the $d$ basis
vectors (barycentric spanner~\cite{AwerbuchKleinberg08JCSS} as
defined in the following) evenly to estimate the qualities of these
vectors. In the exploitation sequence, we select the action
estimated as the best by linearly interpolating the estimated
qualities of the basis vectors. A detailed algorithm is given in
Fig.~\ref{fig:SOSP}.

\begin{definition}
A set $B=\{x_1,\ldots,x_d\}$ is a {\em barycentric spanner} for a
$d$-dimensional set $A$ if every $x\in A$ can be expressed as a
linear combination of elements of $B$ using coefficients
in~$[-1,1]$.
\end{definition}

Note that a compact subset of $\Rc^d$ always has a barycentric
spanner, which can be constructed efficiently (see an algorithm
in~\cite{AwerbuchKleinberg08JCSS}). For the problem at hand, the
path set $\Pc$ is a compact subset of $\Rc^d$ with dimension $d\le
m$. We can thus construct a barycentric spanner consisting of $d$
paths in $\Pc$ by using the algorithm
in~\cite{AwerbuchKleinberg08JCSS}.

\begin{figure}[htbp]
\begin{center}
\noindent\fbox{
\parbox{6in}
{\centerline{\underline{{\bf Adaptive Shortest-Path Routing
Algorithm}}} {
\begin{itemize}
\item Notations and Inputs: Construct a barycentric spanner
$B=\{\vec{p}_1,\vec{p}_2,\ldots,\vec{p}_d\}$ of the vector space of
all paths. Let $\Ac(t)$ denote the set of time indices that belong
to the exploration sequence up to (and including) time $t$ and
$\Ac(1)=\{1\}$. Let $|\Ac(t)|$ denote the cardinality of $\Ac(t)$.
Let $\overline{\theta}_{p_i}(t)$ denote the sample mean of path
$p_i~(i\in\{1,\ldots,d\})$ computed from the past cost observations
on the path. For two positive integers $k$ and $l$, define $k\oslash
l\defeq ((k-1)~\mbox{mod}~l)+1$, which is an integer taking values
from
$1,2,\cdots,l$.\\[-1.5em]
\item At time $t$,
\begin{enumerate}
\item[1.] if $t\in\Ac(t)$, choose path $p_n$ with $n=|\Ac(t)|\oslash d$;
\item[2.] if $t\notin\Ac(t)$, estimate the sample mean of each path
$p_n\in\Pc$ by linearly interpolating $\{\overline{\theta}_{p_1}(t),
\overline{\theta}_{p_2}(t),\ldots,\overline{\theta}_{p_d}(t)\}$.
Specifically, let $\{a_i:|a_i|\le1\}_{i=1}^d$ be the coefficients
such that $\vec{p}_n=\sum_{i=1}^da_i\vec{p}_i$, then
$\overline{\theta}_{p_n}(t)=\sum_{i=1}^da_i\overline{\theta}_{i}(t)$.
Choose path
\[p^*=\arg\min\{\overline{\theta}_{p_n}(t),~n=1,2,\dots,|\Pc|\}.\]
\end{enumerate}
\end{itemize}} }} \caption{The general structure of the algorithm.}\label{fig:SOSP}
\end{center}
\end{figure}

\begin{theorem}\label{thm:opt}
Construct an exploration sequence as follows. Let $a,\zeta,u_0$ be
the constants such that~\eqref{eqn:chernoff} holds on each edge
cost. Choose a constant $b>2m/a$, a constant
{\small\[c\in(0,\min_{j:\mathbb{E}[C_t(\sigma(j))-C_t(\sigma(1))]>0}
\{\mathbb{E}[C_t(\sigma(j))-C_t(\sigma(1))]\}),\]}and a constant
$w\ge\max\{b/(md\zeta u_0)^2, 4b/c^2\}$. For each $t>1$, if
$|\Ac(t-1)|<d\lceil d^2w\log t\rceil$, then include $t$ in $\Ac(t)$.
Under this exploration sequence, the resulting policy $\pi^*$ has
regret {\[R^{\pi^*}(T)\le Amd^3\log T\]}for some constant $A$
independent of $d$, $m$ and $T>1$.
\end{theorem}
\begin{proof}
Since $\Ac(T)\le d\lceil d^2w\log T\rceil$, the regret caused in the
exploration sequence is at the order of $md^3\log T$. Now we
consider the regret caused in the exploitation sequence. Let $E_k$
denote the $k$th exploitation period which is the $k$th contiguous
segment in the exploitation sequence. Let $E_k$ denote the $k$th
exploitation period. Similar to the proof of Theorem~3.1
in~\cite{Liu&Zhao11Allerton}, we have
\begin{eqnarray}\label{eqn:boundE}
|E_k|\le ht_k
\end{eqnarray}for some constant $h$ independent of $d$
and $m$. Let $t_k>1$ denote the starting time of the $k$th
exploitation period. Next, we show that by applying the
Chernoff-Hoeffding bound in~\eqref{eqn:chernoff} on the path cost,
for any $t$ in the $k$th exploitation period and $i=1,\ldots,d$, we
have
\[\Pr(|\mathbb{E}[C_t(p_i)]-\overline{\theta}_{p_i}(t)|\ge c/(2d))\le
2t_k^{-ab/m}.\]To show this, we define the parameter
$\epsilon_i(t)\defeq\sqrt{b\log t/\tau_i(t)}$, where $\tau_i(t)$ is
the number of times that path $i$ has been sampled up to time $t$.
From the definition of parameter $b$, we have
\begin{eqnarray}\label{eqn:neighbor}
\epsilon_i(t)\le\min\{m\zeta u_0,c/(2d)\}.
\end{eqnarray}Applying the
Chernoff-Hoeffding bound, we arrive at
\[\Pr(|\mathbb{E}[C_t(p_i)]-\overline{\theta}_{p_i}(t)|\ge c/(2d))\le
\Pr(|\mathbb{E}[C_t(p_i)]-\overline{\theta}_{p_i}(t)|\ge
\epsilon_i(t))\le 2t_k^{-ab/m}.\]In the exploitation sequence, the
expected times that at least one path in $B$ has a sample mean
deviating from its true mean cost by $c/(2d)$ is thus bounded by
\begin{eqnarray}\label{eqn:boundo}
\Sigma_{k=1}^{\infty}2dt_k^{-ab/m}t_k\le\Sigma_{t=1}^{\infty}2dt^{1-ab/m}=gd.
\end{eqnarray}for some constant $g$ independent of $d$ and $m$. Based on the
property of the barycentric spanner, the best path would not be
selected in the exploitation sequence only if one of the basis
vector in $B$ has a sample mean deviating from its true mean cost by
at least $c/(2d)$. We thus proved the theorem.
\end{proof}

In Theorem~\ref{thm:opt}, we need a lower bound (parameter $c$) on
the difference in the cost means of the best and the second best
paths. We also need to know the bounds on parameters $\zeta$ and
$u_0$ such that the Chernoff-Hoeffding bound~\eqref{eqn:chernoff}
holds. These bounds are required in defining $w$ that specifies the
minimum leading constant of the logarithmic cardinality of the
exploration sequence necessary for identifying the best path.
Similar to~\cite{Liu&Zhao11Allerton}, we can show that without any
knowledge of the cost models, increasing the cardinality of the
exploration sequence of $\pi^*$ by an arbitrarily small amount leads
to a regret linear with $d$ and arbitrarily close to the logarithmic
order with time.

\begin{theorem}\label{thm:nearlog}
Let $f(t)$ be any positive increasing sequence with
$f(t)\rightarrow\infty$ as $t\rightarrow\infty$. Revise policy
$\pi^*$ in Theorem~\ref{thm:opt} as follows: include $t~(t>1)$ in
$\Ac(t)$ if $|\Ac(t-1)|<d\lceil f(t)\log t\rceil$. Under the revised
policy $\pi'$, we have {\[R^{\pi'}(T)=O(df(T)\log T).\]}
\end{theorem}
\begin{proof}
It is sufficient to show that the regret caused in the exploitation
sequence is bounded by $O(d)$, independent of $T$. Since the
exploration sequence is denser than the logarithmic order as in
Theorem~\ref{thm:opt}, it is not difficult to show that the bound on
$|E_k|$ given in~\eqref{eqn:boundE} still holds with a different
value of $h$.

We consider any positive increasing sequence $b(t)$ such that
$b(t)=o(f(t))$ and $b(t)\rightarrow\infty$ as $t\rightarrow\infty$.
By replacing $b$ in the proof of Theorem~\ref{thm:opt} with $b(t)$,
we notice that after some finite time $T_0$, the parameter
$\epsilon_i(t)$ will be small enough to ensure~\eqref{eqn:neighbor}
holds and $b(t)$ will be large enough to ensure~\eqref{eqn:boundo}
holds. The proof thus follows.
\end{proof}

\subsection{Heavy-Tailed Cost Distributions}We now consider the heavy-tailed cost distributions
where the moment of the cost exists up to the $q$-th ($q>1$) order.
From~\cite{Liu&Zhao11Allerton}, we have the following bound on the
deviation of the sample mean from the true mean for heavy-tailed
cost distributions.

Let $\{X(t)\}_{t=1}^{\infty}$ be i.i.d. random variables drawn from
a distribution with the $q$-th moment ($q>1$). Let
$\overline{X}_t=(\Sigma_{k=1}^t X(k))/t$ and
$\theta=\mathbb{E}[X(1)]$. We have, for all $\epsilon>0$,
\begin{eqnarray}\label{eqn:sum1}
\Pr(|\overline{X}_t-\theta|>\epsilon)=o(t^{1-q}).
\end{eqnarray}

\begin{theorem}\label{thm:general}
Construct an exploration sequence as follows. Choose a constant
$v>0$. For each $t>1$, if $|\Ac(t-1)|<vt^{1/q}$, then include $t$ in
$\Ac(t)$. Under this exploration sequence, the resulting policy
$\pi^q$ has regret {\small$R^{\pi^q}(T)\le DdT^{1/q}$} for some
constant $D$ independent of $d$ and $T$.
\end{theorem}
\begin{proof}
Based on the construction of the exploration sequence, it is
sufficient to show that the regret in the exploitation sequence is
$o(T^{1/q})\cdot d$. From~\eqref{eqn:sum1}, we have, for any
$i=1,\ldots,d$,
\[\Pr(|\mathbb{E}[C_t(p_i)]-\overline{\theta}_{p_i}(t)|\ge c/(2d))
=o(|\Ac(t)|^{1-q}).\]For any exploitation slot $t\in\Ac(t)$, we have
$|\Ac(t)|\ge vt{1/q}$. We arrive at
\[\Pr(|\mathbb{E}[C_t(p_i)]-\overline{\theta}_{p_i}(t)|\ge c/(2d))
=o(t^{(1-q)/q}).\]Since the best path will not be chosen only if at
least one of the basis vector has the sample deviating from the true
mean by $c/(2d)$, the regret in the exploitation sequence is thus
bounded by \[\sum_{t=1}^To(t^{(1-q)/q})\cdot d=o(T^{1/q})\cdot d.\]
\end{proof}

\section{Generalization to Stochastic Online Linear Optimization
Problems}\label{sec:SOLO}

In this section, we consider the general Stochastic Online Linear
Optimization (SOLO) problems that include adaptive routing as a
special case. In a SOLO problem, there is a compact set
$\Uc\subset\Rc^d$ with dimension $d$, referred to as the action set.
At each time $t$, we choose a point $x\in \Uc$ and observe a cost
$\vec{C}_t\cdot x$ where $\vec{C}_t\in\Rc^d$ is drawn from an
unknown distribution\footnote{The following result holds in a more
general scenario that only requires the expected cost is linear with
$x$.}. We assume that $\{\vec{C}_t\}_{t\ge1}$ are i.i.d. over $t$.
The objective is to design a sequential action selection policy
$\pi$ to minimize the regret $\Rc^{\pi}(T)$, defined as the total
difference in the expected cost over $T$ slots compared to the
optimal action $x^*\defeq\arg\min\{\mathbb{E}[\vec{C}_t\cdot x]\}$:
{\[\Rc^{\pi}(T)\defeq\mathbb{E}[\sum_{t=1}^T(\vec{C}_t\cdot
x(\pi)-\vec{C}_t\cdot x^*)],\]}where $x(\pi)$ denotes the action
under policy $\pi$ at time $t$.

As a special case, the action space of the adaptive routing problem
consists of all path vectors with dimension $d\le m$. The main
difficulty in a general SOLO problem is on identifying the optimal
action in an infinite action set. Our basic approach is to implement
the shortest-path algorithm in Fig.~\ref{fig:SOSP} repeatedly for
the chosen action to gradually converge to the optimal action.
Specifically, the proposed algorithm runs under an epoch structure
where each epoch simulates one round of the algorithm and the epoch
lengths form a geometric progression.

\begin{theorem}\label{thm:general}
Assume that the cost for each action is light-tailed. Let $T_k$ be
the length of the $k$-th epoch. Construct an exploration sequence in
this epoch as follows. Let $a,\zeta,u_0$ be the constants such
that~\eqref{eqn:chernoff} holds on each cost. Choose a constant
$b>2/a$, a constant $c=\log T_k^{1/3}/T_k^{1/3}$, and a constant
$w\ge\max\{b/(d\zeta u_0)^2, 4b/c^2\}$. For each $t>1$, if
$|\Ac(t-1)|<d\lceil d^2w\log t\rceil$, then include $t$ in $\Ac(t)$.
The resulting policy $\pi^g$ has regret
{\[R^{\pi^g}(T)=O(d^3T^{2/3}\log^{1/3}T).\]}
\end{theorem}We point out that the regret order can
be improved to linear with $d$ by sacrificing the order with $T$ by
an arbitrarily small amount, as similar to
Theorem~\ref{thm:nearlog}.

\section{Conclusion}

In this paper, we considered the adaptive routing problem in
networks with unknown and stochastically varying link states, where
only the total end-to-end cost of a path is observable after the
path is selected for routing. For both light-tailed and heavy-tailed
link-state distributions, we proposed efficient online learning
algorithms to minimize the regret in terms of both time and network
size. The result was further extended to the general stochastic
online liner optimization problems. Future work includes extending
the i.i.d. cost evolution over time to more general stochastic
processes.

\end{document}